\documentclass[clean]{eptcs-clean}

\usepackage{microtype}

\usepackage{amsmath,amsthm,amssymb}
\usepackage{xcolor}

\usepackage{tikz}

\usepackage{nicefrac}

\usepackage{thmtools,thm-restate} 

\theoremstyle{definition}
\newtheorem{defn}{Definition}

\theoremstyle{plain}

\newtheorem{lemma}[defn]{Lemma}

\theoremstyle{remark}

\theoremstyle{definition}

\newcommand{\Neumark}{Na\u{\i}\-mark}

\newcommand{\tuple}[1]{\mathopen{\langle}#1\mathclose{\rangle}} 
\usepackage{colonequals}
\newcommand{\defeq}{\colonequals} 

\newcommand{\setdef}  [2]{\mathopen{}\left\{#1 \mid #2\right\}\mathclose{}} 
\newcommand{\enset}   [1]{\mathopen{}\left\{#1\right\}\mathclose{}}         
\newcommand{\fdec}    [3]{#1\colon #2 \longrightarrow #3}

\newcommand{\fdecdef} [5]{#1\colon #2 \longrightarrow #3 \coloncolon #4 \longmapsto #5}

\newcommand{\ZZ}  {\mathbb{Z}}

\newcommand{\CC}  {\mathbb{C}}

\newcommand{\Forall}[1]{\forall {#1}\boldsymbol{.}\;}

\def\bimplies{\Leftrightarrow}
\def\implies{\Rightarrow}

\newcommand{\Mcomma}{\text{,}}

\newcommand{\Mdot}{\text{.}}
\newcommand{\Mand}{\quad\quad\text{ and }\quad\quad}

\DeclareMathOperator{\Tr}{Tr} 

\newcommand{\M}{\mathcal{M}}

\newcommand*\Distr{\mathcal{D}}

\newcommand{\BellState}{\ket{\Phi^+}}

\newcommand{\GHZt}{\mathsf{GHZ}}
\newcommand{\GHZ}{\ket{\GHZt}}

\newcommand{\GHZnt}{\mathsf{GHZ}(n)}
\newcommand{\GHZn}{\ket{\GHZnt}}

\newcommand{\Wt}{\mathsf{W}}
\newcommand{\W}{\ket{\Wt}}

\newcommand{\WSLOCCt}{\psi_{_{\mathsf{W}}}}
\newcommand{\GHZSLOCCt}{\psi_{_\mathsf{GHZ}}}
\newcommand{\WSLOCC}{\ket{\WSLOCCt}}
\newcommand{\GHZSLOCC}{\ket{\GHZSLOCCt}}
\newcommand{\vlambda}{\boldsymbol{\lambda}}
\newcommand{\Balt}{\mathsf{B}_{\vlambda,\Phi}}
\newcommand{\Bal}{\ket{\Balt}}
\newcommand{\Balzerot}{\mathsf{B}_{\vlambda,0}}
\newcommand{\Balzero}{\ket{\Balzerot}}
\newcommand{\Balargt}[2]{\mathsf{B}_{#1,#2}}
\newcommand{\Balarg}[2]{\ket{\Balargt{#1}{#2}}}

\newcommand{\ket}[1]{|#1\rangle}
\newcommand{\braket}[2]{\langle#1|#2\rangle}

\newcommand{\kp}{\ket{\psi}}

\newcommand{\vect}[1]{\mathbf{#1}}

\newcommand{\vm}{\vect{m}}
\newcommand{\vo}{\vect{o}}

\newcommand{\vtheta}{\boldsymbol{\theta}}
\newcommand{\vphi}{\boldsymbol{\varphi}}

\newcommand{\bl}[2]{\ket{#1,#2}}

\newcommand{\blTP}{\bl{\theta}{\varphi}}

\newcommand{\blTPone}{\bl{\theta_1}{\varphi_1}}
\newcommand{\blTPn}  {\bl{\theta_n}{\varphi_n}}

\newcommand{\blTPv}{\bl{\vtheta}{\vphi}}

\newcommand{\blout}[3]{\ket{#1,#2 \mapsto #3}}
\newcommand{\blTPO}{\blout{\theta}{\varphi}{o}}
\newcommand{\blTPOO}[1]{\blout{\theta}{\varphi}{#1}}

\newcommand{\blTPOone}{\blout{\theta_1}{\varphi_1}{o_1}}
\newcommand{\blTPOn}  {\blout{\theta_n}{\varphi_n}{o_n}}
\newcommand{\blTPOv}{\blout{\vtheta}{\vphi}{\vo}}
\newcommand{\blTPOOv}[1]{\blout{\vtheta}{\vphi}{#1}}

\newcommand{\braketTPv}[1]{\braket{\vtheta,\vphi}{#1}}

\newcommand{\braketTPOv}[1]{\braket{\vtheta,\vphi\mapsto\vo}{#1}}
\newcommand{\braketTPOi}[1]{\braket{\theta_i,\varphi_i \mapsto o_i}{#1}}
\newcommand{\braketTPOOv}[2]{\braket{\vtheta,\vphi\mapsto#1}{#2}}

\newcommand{\TP}{(\theta,\varphi)}
\newcommand{\TPv}{(\vtheta,\vphi)}
\newcommand{\TPone}{(\theta_1,\varphi_1)}
\newcommand{\TPn}  {(\theta_n,\varphi_n)}

\newcommand{\Prob}{\mathsf{Prob}}

\newcommand{\LM}{\mathsf{LM}}

\title{Minimum quantum resources for strong non-locality}
\author{Samson Abramsky \quad Rui Soares Barbosa \quad Giovanni Car\`{u}\
\institute{Department of Computer Science\\
University of Oxford}
\email{\{samson.abramsky, rui.soares.barbosa, giovanni.caru\}@cs.ox.ac.uk}\\
\and
Nadish de Silva
\institute{Department of Computer Science\\
University College London}
\email{nadish.desilva@utoronto.ca}
\and
Kohei Kishida,
\institute{Department of Computer Science\\
University of Oxford}
\email{kohei.kishida@cs.ox.ac.uk}
\and
Shane Mansfield
\institute{School of Informatics\\
University of Edinburgh}
\email{smansfie@staffmail.ed.ac.uk}
}

\def\titlerunning{Minimum quantum resources for strong non-locality}
\def\authorrunning{S. Abramsky, R. S. Barbosa, G. Car{\`u}, N. de Silva, K. Kishida, S. Mansfield}

\begin{document}

\maketitle

\begin{abstract}
We analyse the minimum quantum resources needed to realise strong non-locality,
as exemplified e.g.~by the classical GHZ construction.
It was already known that no two-qubit system,
with any finite number of local measurements, can realise strong non-locality.
For three-qubit systems, we show that strong non-locality can only be realised
in the GHZ SLOCC class, and with equatorial measurements.
However, we show that in this class there is an infinite family of states which are pairwise non LU-equivalent 
that realise strong non-locality with finitely many measurements. These states 
have decreasing entanglement between one qubit and the other two,
necessitating an increasing number of local measurements on the latter.
\end{abstract}

\section{Introduction}\label{sec:intro}

In this paper, we aim at identifying the minimum quantum resources needed to witness strong contextuality \cite{AbramskyBrandenburger}, and more specifically, strong (or maximal) non-locality.
Non-locality is, of course, a fundamental phenomenon in quantum mechanics --
both from a foundational point of view, and with respect to quantum information and computation, in which it plays a central r\^ole.

The original form of Bell's argument \cite{Bell-thm},
as well as its now more standard formulation due to Clauser, Horne, Shimony, and Holt (CHSH) \cite{CHSH},
rests on deriving an inequality
that must be satisfied by probabilities arising from any local realistic theory,
but which is violated by those predicted by quantum mechanics
for a particular choice of a state and a finite set of measurements.
Greenberger, Horne, Shimony, and Zeilinger (GHSZ) \cite{GHZ,GHSZ90}
gave a stronger, inequality-free argument for quantum non-locality.
This depended only on the \emph{possibilistic} aspects of quantum predictions,
i.e.~on which joint outcomes given a choice of measurements have non-zero probability,
regardless of the actual value of the probabilities.
Their argument was later simplified by Mermin \cite{Mermin90:QuantumMysteriesRevisited-SimplifiedGHZ1,Mermin90:SimpleUnifiedForm}.
Whereas the Bell--CHSH argument used local measurements on a two-qubit system prepared in a maximally entangled state,
the GHZ--Mermin argument required a three-qubit system in the GHZ state.
Subsequently, Hardy  showed that
one can indeed find a proof of non-locality ``without inequalities'', i.e. based on possibilistic information alone,
using a bipartite, two-qubit system \cite{Hardy92:nonlocality1}.
Hardy's argument works on any two-qubit entangled state
bar the maximally entangled ones \cite{Hardy93:nonlocality2}.
In fact, a similar argument works on almost all $n$-qubit states \cite{AbramskyConstantinYing2015:HardyIsAlmostEverywhere}, the exceptions being those states which are products of one-qubit states and two-qubit maximally entangled states,
which provably do \emph{not} admit any non-locality argument ``without inequalities'' \cite{Mansfield17:Hardy}.
However, 
there is an important logical distinction between the GHSZ and Hardy possibilistic arguments.

Abramsky and Brandenburger \cite{AbramskyBrandenburger} introduced
a general mathematical framework for contextuality,
in which non-locality arises as a particular case.
This approach studies these phenomena at a level of generality
that abstracts away from the particularities of quantum theory.
The point is that contextuality and non-locality are witnessed by the empirical data itself, without presupposing any physical theory.
For this reason, one deals with ``empirical models'' --
tables of data for a given experimental scenario, obtained from empirical observations or predicted by some physical theory,
specifying probabilities of joint outcomes for the allowed sets of compatible measurements.

Various kinds of contextuality (or, in particular, non-locality) arguments were studied and classified at this abstract level,
leading to the introduction of a qualitative hierarchy of strengths of contextuality in \cite{AbramskyBrandenburger},
with further refinements in \cite{AbramskyMansfieldBarbosa:Cohomology-QPL,AbramskyEtAl:ContextualityCohomologyAndParadox}.
The classic arguments for quantum non-locality, familiar from the literature,
sit at different levels in this hierarchy.
There is a strict relationship of strengths of non-locality, rendered as
\[\text{Bell} < \text{Hardy} < \text{GHZ} \Mcomma\]
where these representative examples correspond, respectively, to
probabilistic non-locality, possibilistic  non-locality, and strong  non-locality.

Strong contextuality (or, in particular, non-locality)  arises when there is \emph{no assignment of outcomes to all the measurements} consistent with the events that the empirical model deems possible, i.e.~to which it attributes non-zero probability. It is exactly this impossibility which is shown by Mermin's classic argument in \cite{Mermin90:QuantumMysteriesRevisited-SimplifiedGHZ1}.
Strong contextuality is also the highest level of contextuality in a different, quantitative sense.
It turns out to coincide with the notion of maximal  contextuality,
the property that an empirical model admits no proper decomposition into a convex combination of a non-contextual model and another model.
This corresponds to attaining the maximum value of $1$ for the \emph{contextual faction}, a natural measure of contextuality introduced in 
\cite{AbramskyBrandenburger} as a generalisation of the notion of non-local fraction \cite{ElitzurPopescuRohrlich1992:QuantumNonlocalityForEachPairInAnEnsemble,BarrettKentPironio2006:MaximallyNonlocalAndMonogamousQuantumCorrelations,AolitaEtAl2012:FullyNonlocalQuantumCorrelations}. The contextual fraction is shown in \cite{AbramskyBarbosaMansfield17:CF} to be equal to the \emph{maximal normalised violation} of a contextuality-witnessing inequality. Hence, a model is strongly contextual if and only if it violates a generalised Bell inequality up to its algebraic bound.

Strong non-locality is particularly relevant to quantum computing.
It is exhibited, for example, by all graph states under stabiliser measurements \cite{GuhneEtAl:BellInequalitiesForGraphStates},
which provide resource states and measurements for universal quantum computing via the one-way or measurement-based model \cite{RaussendorfBriegel01:OneWay}.
It is also known to be necessary for increasing computational power in certain models of measurement-based quantum computing with restricted classical co-processing \cite{RaussendorfSC}.
For instance, in \cite{AndersBrowne} it was shown that GHZ strong non-locality enables a linear classical co-processor to implement the non-linear $\mathsf{AND}$ function, and subsequently in \cite{DunjkoKapourniotisKashefi:QuantumEnhancedSecureDelegatedClassicallComputing} that it enables the function to be implemented in a secure delegated way.
Moreover, strong non-locality has important consequences for certain information processing tasks: in particular, it is known to be required for perfect strategies \cite{MancinskaRobersonVarvitsiotis16:DecidingExistencePerfectStrategies} in certain cooperative games \cite{AbramskyBarbosaMansfield17:CF}.

\subsection*{Summary of results}

In this paper, our aim is to analyse the minimum quantum resources needed to realise strong non-locality. More precisely, we consider $n$-qubit systems viewed as $n$-partite systems,\footnote{We know by a result of Heywood and Redhead \cite{HeywoodRedhead83:Nonlocality} that strong contextuality can be realised  using a bipartite system, but with a qutrit at each site. Hence our focus on qubits.}
where each party can perform one-qubit local projective measurements.\footnote{Throughout
this paper, we focus on projective measurements.
The more general POVMs are justified as physical processes by \Neumark's dilation,
since they are described as projective measurements in a larger physical system.
Given that we are interested in characterising the minimum resources needed in order to witness strong non-locality, it seems reasonable to focus on PVMs, which do not need to be seen as measurements on a part of a larger system.
}
We shall consider the case where each party has a finite set of measurements available --
this is what corresponds to the standard experimental scenarios for non-locality.

\begin{itemize}
\item The first result we present is limitative in character.
It shows that strong non-locality \emph{cannot} be realised by a two-qubit system with any finite number of local measurements.
This result was already proven, using different terminology,
in \cite{BrassardMethotTapp2005:MinimumEntangledStatePseudoTelepathy}.
However, we include it for completeness and because its proof
is useful as a warm-up for proving the other results in this paper.\footnote{Note that, in the same paper, it is also shown that the result applies to any bipartite state where one of the systems is a qubit, by an application of Schmidt decomposition of any bipartite state. This means that the optimal dimention in which strong non-locality can be realised is $2 \times 2 \times 2 = 8$, i.e. a three-qubit system, since a two-qutrit system has dimension $9$.}

There is a subtle counterpoint to this in a result from \cite{BarrettKentPironio2006:MaximallyNonlocalAndMonogamousQuantumCorrelations}, which shows that using a maximally entangled bipartite state,
and an infinite family of local measurements, strong non-locality is achieved ``in the limit'' in a suitable sense. More precisely, as more and more measurements from the family are used, the local fraction -- the part of the behaviour which can be accounted for by a local model -- tends to $0$, or equivalently the non-local fraction tends to $1$.
There is an interesting connection to this in our results for the tripartite case.

However, there is a practical advantage in being able to witness strong non-locality with a fixed finite number of measurements.
If one wishes to design an experimental test for maximal non-locality,
it is desirable that one can increase precision, i.e. increase the lower bound on the non-local fraction,
without needing to expand the experimental setup -- in particular, the number of measurement settings required to be performed --
but rather by simply performing more runs of the same experiment.

\item Having shown that strong non-locality cannot be realised in the two-qubit case, we turn to the analysis of three-qubit systems. Of course, we know by the classical GHSZ--Mermin construction that strong non-locality \emph{can} be achieved in this case, using the GHZ state and Pauli  $X$ and $Y$ measurements on each of the qubits. Our aim is to analyse for which states, and with respect to which measurements, can strong non-locality be achieved.
We use the classification into SLOCC classes for tripartite qubit systems from \cite{DurVidalCirac00:3-SLOCC}.
According to this analysis, there are two maximal SLOCC classes, the GHZ and W classes. Below these, there are the degenerate cases of products of an entangled bipartite state with a one-qubit state, e.g.~$AB{-}C$. By the previous result, these degenerate cases cannot realise strong non-locality. We furthermore show that no state in the W class can realise strong non-locality, for any choice of finitely-many local measurements.

\item This leaves us with the GHZ SLOCC class. We use the detailed description of this class as a parameterised family of states from \cite{DurVidalCirac00:3-SLOCC}. 
We first show that any state in this class witnessing strong non-locality 
with finitely many local measurements must satisfy a number of constraints on the parameters.
In particular, the state must be balanced
in the sense that the coefficients in its unique linear decomposition into a pair of product states have the same complex modulus.
We furthermore show that \emph{only equatorial measurements need be considered} (the equators being uniquely determined by the state) -- no other measurements can contribute to a strong non-locality argument.

\item Having thus narrowed the possibilities for realising strong non-locality considerably, we find a new infinite family of models displaying strong non-locality using states within the GHZ SLOCC class that are not LU-equivalent to the GHZ state.
The states in this family start from GHZ and tend in the limit to the state $\BellState \otimes \ket{+}$ in the AB--C class with maximal entanglement on the first two qubits,  and in product with the third.
This family is actually closely related to the construction from
\cite{BarrettKentPironio2006:MaximallyNonlocalAndMonogamousQuantumCorrelations}
in which an increasing number of measurements on a bipartite maximally entangled state
eventually squeezes the local fraction to zero in the limit.
Our family is obtained by adding a third qubit to this setup, with two available local measurements,
and some entanglement between the first two qubits and the third one,
thus allowing  strong non-locality to be witnessed with a finite number of measurements.
There is a trade-off between the number of measurement settings available on the first two qubits -- and, consequently, the lower bound for the non-local fraction these measurements can witness -- and the amount of entanglement necessary between the third qubit and the original two.
\end{itemize}

\subparagraph*{Outline.}
The remainder of this article is organised as follows:
Section~\ref{sec:background} summarises some background material on non-locality and entanglement classification of three-qubit states,
Section~\ref{sec:bipartite} shows that strong non-locality cannot be witnessed by two-qubit states and a finite number of local measurements;
Section~\ref{sec:W} does the same for three-qubit states in the SLOCC class of W;
Section~\ref{sec:GHZ} deals with states in the SLOCC class of GHZ, deriving conditions on these necessary for strong non-locality;
Section~\ref{sec:GHZ-family} presents the family of strong non-locality arguments using states in the GHZ-SLOCC class;
and Section~\ref{sec:outlook} concludes with some discussion of open problems and further directions.
Detailed proofs of all the results are found in the Appendix.

\section{Background}\label{sec:background}

\subsection{Measurement scenarios and empirical models}

We summarise some of the main ideas of \cite{AbramskyBrandenburger}, with particular emphasis on non-locality. This is merely an instance of contextuality in a particular kind of measurement scenarios known as multipartite Bell-type scenarios.
For each notion, we introduce the general definition followed by its specialisation to multipartite Bell-type scenarios.

Measurement scenarios are abstract descriptions of experimental setups.
In general, a \emph{measurement scenario} is described by
a set
of measurement labels $X$, a set of outcomes $O$, and a cover $\M$ of $X$ consisting of measurement contexts, i.e. maximal sets of measurements that can be jointly performed. We are typically interested in measurement scenarios with finite $X$, but for technical reasons it will be useful to consider scenarios with infinitely many measurements in order to prove results about all their finite `subscenarios' at once.
Throughout this paper, we shall also restrict our attention to dichotomic measurements, with outcome set $O = \enset{-1,+1}$.
This is a reasonable restriction, especially since our main focus shall be projective measurements on single qubits.
Multipartite Bell-type scenarios are a particular kind of measurement scenario which can be thought to describe multiple parties at different sites, each independently choosing to perform one of a number of measurements available to them.
More formally,
an $n$-partite Bell-type scenario is described by
sets $X_1, \ldots, X_n$ labelling the measurements available at each site (so that $X \defeq X_1 \sqcup \cdots \sqcup X_n$),
with maximal contexts corresponding to a single choice of measurement for each party,
or in other words a tuple 
$\vm = \tuple{m_1, \ldots, m_n} \in X_1 \times \cdots \times X_n$
(so $\M \cong \prod_{i=1}^n X_i$).

An empirical model is a collection of probabilistic data
representing possible results of running the experiment represented by a measurement scenario.
Given a measurement scenario $\tuple{X,\M,O}$, an \emph{empirical model} on that scenario is a family 
$\enset{e_C}_{C \in \M}$ where each $e_C \in \Distr(O^C)$
is a distribution over the set of joint outcomes to the measurements of $C$.
Given an assignment $\fdec{s}{C}{O}$ of outcomes to each measurement in $C$, 
the value $e_C(s)$ is the probability of obtaining the outcomes determined by $s$ when jointly performing the measurements in the context $C$. 

In the particular case of a Bell-type scenario, we have a family
$\enset{e_\vect{m} \in \Distr(O^n)}_{\vect{m} \in \prod_i X_i}$
of probability distributions.
Given a vector of outcomes $\vo = \tuple{o_1,\ldots,o_n} \in O^n$,
the probability $e_\vm(\vo)$ of obtaining the joint outcomes $\vo$ upon performing the measurements $\vm$ at each site
is often denoted in the literature on non-locality as follows:
\[e_\vm(\vo) =  \Prob(\vo | \vm) = \Prob(o_1, \ldots, o_n | m_1, \ldots, m_n) \Mdot\]

Empirical models are usually assumed to satisfy a \emph{compatibility} condition:
that marginal distributions agree on overlapping contexts, i.e. for all $C$ and $C'$ in $\M$, $e_C|_{C \cap C'} = e_{C'}|_{C \cap C'}$.
In the case of multipartite scenarios, this corresponds to the familiar \emph{no-signalling} condition.

\subsection{Contextuality and non-locality}
An empirical model is said to be \emph{non-contextual} if
there is a distribution on 
assignments of outcomes to all the measurements, $d \in \Distr(O^X)$,
that marginalises to the empirical probabilities for each context, i.e. $\Forall{C \in \M} d|_C = e_C$.
Note that this means there is a deterministic, non-contextual hidden-variable theory with the set of global assignments $O^X$
serving as a canonical hidden variable space.
Indeed, the existence of such a global distribution is in fact equivalent to the existence of a probabilistic hidden variable theory that is factorisable,
a notion that in multipartite scenarios specialises to the standard formulation of Bell locality:
there is a set of hidden variables $\Lambda$, a distribution in $h\in\Distr(\Lambda)$, and ontic probabilities
$\Prob(\vo | \vm, \lambda)$ that are consistent with
the empirical ones, i.e. for all $\vm \in \M$ and $\vo \in O_n$
\[
\sum_{\lambda\in\Lambda}\Prob(\vo | \vm, \lambda)h(\lambda)
=
\Prob(\vo | \vm)
=
e_\vm(\vo)
\Mcomma\] 
and that factorise when conditioned on each $\lambda \in \Lambda$, i.e.
\[
\Prob(\vo | \vm, \lambda)
=
\prod_{i=1}^n \Prob(o_i | m_i,\lambda)
\Mdot\]
where the probabilities on the right-hand side are obtained as the obvious marginals.
The equivalence between the two formulations of non-contextuality or locality --
in terms of a probability distribution on global assignments (canonical deterministic hidden variable theory)
and in terms of factorisable hidden variable theory --
was proven in \cite{AbramskyBrandenburger} for general measurement scenarios,
vastly extending a result by Fine \cite{Fine82}.
This justifies viewing non-locality as the special case of contextuality in multipartite systems.

For some empirical models, it suffices to consider their possibilistic content,
i.e. whether events are possible (non-zero probability) or impossible (zero probability),
to detect the presence of contextuality.
In this case, we say that the model is \emph{logically contextual}.
An even stronger form of contextuality, which will be our main concern in this article,
arises when
no global assignment of outcomes to all measurements is consistent with the events deemed possible by the model:
the empirical model $e$ is said to be \emph{strongly contextual} if there is no assignment $\fdec{g}{X}{O}$ such that
$\Forall{C \in \M} e_C(g|_C) > 0$.
In the particular case of multipartite scenarios,
such a global assignment is determined by a family of maps $\fdec{g_i}{X_i}{O}$ for each site $i$ so that $\fdec{g = \bigsqcup_{i=1}^n g_i}{\bigsqcup_{i=1}^n X_i}{O}$. The consistency condition then reads:
for any choice of measurements
$\vm = \tuple{m_1, \ldots, m_n} \in  \prod X_i$,
writing $g(\vm) = \tuple{g_1(m_1),\ldots,g_n(m_n)}$, we have
\[e_\vm(g(\vm)) = \Prob(g(\vm) | \vm) = \Prob(g_1(m_1),\ldots,g_n(m_n) | m_1, \ldots, m_n) > 0 \Mdot\]

As mentioned in Section~\ref{sec:intro},
strong contextuality was shown in \cite{AbramskyBrandenburger} to exactly capture the notion of maximal contextuality.
The proof of this equivalence depends crucially on the finiteness of the number of measurements.
If one would consider an infinite number of measurements,
a situation could occur in which there is a global assignment $g$ consistent with the model,
in the sense that $\Forall{C \in \M} e_C(g|_C)>0$, but where $\inf_{C\in \M}{e_C(g|_C)} = 0$,
in which case $g$ does not correspond to any positive fraction of the model.
This will indeed be the case for all the consistent global assignments described in this paper.
Note, however, that proving the failure of strong contextuality
in a scenario with an infinite number of measurements,
even if the witnessing global assignment has  $\inf_{C \in \M}e_C(g|C) = 0$,
is nonetheless sufficient to show that maximal contextuality cannot be realised using only a finite subset of the measurements.

\subsection{Quantum realisable models}\label{ssec:quantum}
We are mainly concerned with empirical models that are realisable by quantum systems.
This means that one can find a
quantum state and associate to each measurement label a quantum measurement in the same Hilbert space such that 
measurements in the same context commute and the probabilities of the various outcomes
are given by the Born rule.

More specifically, we are concerned with models arising from $n$-qubit systems with local, i.e. single-qubit, measurements.
The Bloch sphere representation of one-qubit pure states will be useful:
assuming a preferred orthonormal basis $\enset{\ket{0},\ket{1}}$ of $\CC^2$,
we shall use the notation
\[
\blTP \defeq \cos{\frac{\theta}{2}} \ket{0} +  e^{i\varphi}\sin{\frac{\theta}{2}} \ket{1}
\]
for any $\theta \in [0,\pi]$ and $\varphi \in [0, 2\pi)$.

Any single-qubit projective measurement is fully determined by specifying such a normalised vector in $\CC^2$,
namely the pure state corresponding to the $+1$ eigenvalue or outcome.
Hence, the set of local measurements for a single qubit is labelled by 
\[\LM = [0,\pi] \times [0, 2\pi)\]
The quantum measurement determined by $(\theta,\varphi) \in \LM$
has eigenvalues $O = \enset{+1,-1}$ with the eigenvector corresponding to outcome $o \in O$ given by:
\[
\blTPO \defeq
\begin{cases}
\blTP
& \text{if $o = +1$} \\
\bl{\pi-\theta}{\varphi + \pi}
& \text{if $o = -1$}
\end{cases}
\]

Throughout this paper, we shall be considering the $n$-partite measurement scenario with $X_i = \LM$ for every site.
Measurement contexts correspond to a choice of single qubit measurements for each of the $n$ sites, represented by a tuple
$\TPv = \tuple{\TPone,\ldots,\TPn}$.
Performing all the measurements of a context in parallel yields an outcome $\vo = \tuple{o_1, \ldots, o_n} \in O^n$.
The vector corresponding to this outcome is denoted  
\[
\blTPOv \defeq \blTPOone \otimes \cdots \otimes \blTPOn \Mdot
\]
We shall also find it useful to write 
\[
\blTPv \defeq  \blTPone \otimes \cdots \otimes \blTPn = \blTPOOv{\tuple{+1, \ldots, +1}}
\]
for the vector corresponding to the joint outcome assigning $+1$ at every site.

An $n$-qubit state $\ket{\psi}$
determines an empirical model $e^{\ket{\psi}}$ for this measurement scenario:
\[e^{\ket{\psi}}_{\TPv}(\vo) = \Prob^{\ket{\psi}}(o_1,\ldots,o_n | \TPone,\ldots,\TPn) \defeq |\braketTPOv{\psi}|^2 \Mdot\]
We are concerned with checking for strongly non-local behaviour on such a model.
As explained in the previous section, this amounts to checking for the existence of maps
$\fdec{g_i}{\LM}{O}$ for each site such that for any choice of measurements $\TPv$,
the corresponding outcome has positive probability:
\begin{align*}
e_{\TPv}(g\TPv) &= \Prob^{\ket{\psi}}(g_1\TPone,\ldots, g_n\TPn | \TPone,\ldots,\TPn) \\
&= |\braketTPOOv{g\TPv}{\psi}|^2 > 0 \Mdot
\end{align*}
Given that these are quantum probabilities, we can rephrase this condition in terms of non-vanishing amplitudes:
$\braketTPOOv{g\TPv}{\psi} \neq 0$.

The following fact will be used throughout. Suppose we want to check the consistency with the empirical model of a given
global assignment $g = \bigsqcup_{i=1}^n g_i$. If this assignment satisfies
\begin{equation}\label{eq:g-negpreserve}
\Forall{i \in \enset{1, \ldots,n}}  g_i(\theta,\varphi) = -g_i(\pi - \theta,\varphi+\pi) \Mcomma
\end{equation}
that is, measurements with $+1$ eigenstates diametrically opposed in the Bloch spehere (i.e. measurements that are the negation of each other) are assigned opposite outcomes,
then
\[
\blTPOO{g_i\TP} =
\begin{cases}
\blTP
& \text{if $g_i\TP = +1$} \\
\bl{\pi-\theta}{\varphi + \pi}
& \text{if $g_i\TP = -1 \quad (\bimplies g_i(\pi - \theta,\varphi+\theta) = +1)$}
\end{cases}
\]
meaning that $\blTPOOv{g\TPv} = \bl{\vtheta'}{\vphi'}$ with $g_i(\theta'_i,\varphi_i') = +1$ for all $i$.
In other words, 
should we wish to calculate the amplitude for a joint outcome $\vo$ on a given context $\TPv$,
we may equivalently calculate the amplitude for the joint outcome $\tuple{+1,\ldots,+1}$ on a new context $(\vtheta',\vphi')$ obtained by substituting $\theta_i \mapsto \pi - \theta_i$ and $\varphi_i \mapsto \pi + \varphi_i$ for all $i$ such that $o_i = -1$.
Therefore, it suffices to verify the equation $\braketTPOOv{g\TPv}{\psi} \neq 0$ for all contexts whose measurements are all assigned $+1$.
Indeed, the same is true if \eqref{eq:g-negpreserve} is relaxed to simply say that $g_i(\pi - \theta,\varphi+\pi)=-1 \implies g_i(\theta,\varphi)=+1$.
Incidentally, even though we shall not need this fact,
note that if there is any global assignment consistent with the model, there will be one that satisfies \eqref{eq:g-negpreserve},
for this would only require a subset of the conditions.

We conclude this subsection with two observations regarding these particular quantum empirical models.
First, note that local unitaries (LU) on the state don't affect non-locality, or indeed strong non-locality, of the resulting empirical model.
This follows from the fact that by moving from the Schr\"odinger to the Heisenberg picture, we may equivalently leave the state fixed and apply the corresponding unitaries to the sets of available local measurements.
Since the available local measurements are all the projective one-qubit measurements, a local unitary, which can be seen as a rotation of the Bloch sphere, merely maps this set to itself.
Secondly, if we are dealing with a product state of $n$-qubits, $\ket{\psi} = \ket{\psi_1} \otimes \cdots \otimes \ket{\psi_n}$,
then the resulting empirical model is necessarily local. This is because the
probabilities factorise:
\[
\Prob^{\ket{\psi}}(\vo | \TPv)
=
\left|\braketTPOv{\psi}\right|^2
=
\left|\prod_{i=1}^n\braketTPOi{\psi_i}\right|^2
=
\prod_{i=1}^n\left|\braketTPOi{\psi_i}\right|^2
\Mdot
\]

\subsection{SLOCC classes of three-qubit states}

A classification of multipartite quantum states by their degree of entanglement is given by the notion of LOCC (local operations and classical communication) equivalence \cite{BennettEtAl96:ConcentratingLO,Nielsen99:EntanglementTransformations,KentLindenMassar99:LOCC}.  A protocol is said to be LOCC if it is of the following form: each party may perform local measurements and transformations on their system, and may communicate measurement outcomes to the other parties, so that local operations may be conditioned on measurement outcomes anywhere in the system.   A state $\ket{\psi_1}$ is LOCC-convertible to a state $\ket{\psi_2}$ if there exists a LOCC protocol that \emph{deterministically} produces $\ket{\psi_2}$ when starting with $\ket{\psi_1}$.  Intuitively, such a protocol cannot increase the degree of entanglement and so we think of $\ket{\psi_1}$ as being at least as entangled as $\ket{\psi_2}$.  The notion of LOCC-convertibility defines a preorder%
\footnote{A preorder is a reflexive and transitive relation;
i.e. it is like a partial order except that it can deem two distinct elements equivalent.}
on multipartite states that in turn yields a notion of LOCC-equivalence of states: the states $\ket{\psi}$ and $\ket{\phi}$ are LOCC-equivalent when $\ket{\psi}$ is LOCC-convertible to $\ket{\phi}$ and \textit{vice versa}.  The LOCC-convertibility preorder then naturally defines a partial order on the collection of LOCC equivalence classes of states.

A coarser classification of multipartite quantum states is given by relaxing the requirement that our conversion protocols succeed deterministically to the requirement that they succeed with non-zero probability \cite{BennettEtAl00:MeasuresMultipartiteEntanglement}.  The previous paragraph holds true for SLOCC (stochastic LOCC) \emph{mutatis mutandis}.
Note that equivalence of two states under LU transformations implies their SLOCC-equivalence.
More generally, two states are SLOCC-equivalent if and only if they are related by an invertible local operator (ILO) \cite{DurVidalCirac00:3-SLOCC}.

D\"ur, Vidal, and Cirac \cite{DurVidalCirac00:3-SLOCC} classified the SLOCC classes of three-qubit systems and found there to be exactly six classes (see Figure~\ref{fig:hasse}).  The GHZ and W states are representatives of the two maximal, non-comparable classes.  Three intermediate classes are characterised by bipartite entanglement between two of the qubits, which are in a product with the third.  Finally, the minimal class is given by product states.

By the last observation in the previous section, it is obvious that a state in the A--B--C class cannot realise non-locality,
and that the case of a state in one of the intermediate classes can be reduced to that of the two qubits that are entangled.
Hence, we shall first discuss strong non-locality for two-qubit states and then proceed in turn to each of the maximal SLOCC classes of three-qubit states, W and GHZ.
\begin{figure}
\begin{center}
\begin{tikzpicture}
    \node (ghz) at (-1,0) {GHZ};
    \node (w) at (1,0) {W};
    
    \node (a)  at (-2,-1)  {A--BC};
        \node (b)  at (0,-1)  {B--AC};
\node (c) at (2,-1) {C--AB};

    \node (p) at (0,-2) {A--B--C};

\draw [thick, ->] (ghz) -- (a);
\draw [thick, ->] (ghz) -- (b);
\draw [thick, ->] (ghz) -- (c);

\draw [thick, ->] (w) -- (a);
\draw [thick, ->] (w) -- (b);
\draw [thick, ->] (w) -- (c);

\draw [thick, ->] (a) -- (p);
\draw [thick, ->] (b) -- (p);
\draw [thick, ->] (c) -- (p);
\end{tikzpicture}
\end{center}
\caption{\label{fig:hasse} Hasse diagram of the partial order of three-qubit SLOCC classes.}
\end{figure}
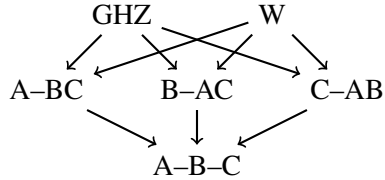

\section{Two-qubit states are not strongly non-local}\label{sec:bipartite}

Every two-qubit state can be written, up to LU, uniquely as 
\begin{equation}\label{equ: bipartite}
\kp=\cos\delta\ket{00}+\sin\delta\ket{11},
\end{equation}
where $\delta\in [0,\frac{\pi}{4}]$. The state \eqref{equ: bipartite} is either: the product state $\ket{00}$, which is obviously non-contextual since it is separable, when $\delta=0$; or an entangled state in the SLOCC class of the Bell state $\ket{\Phi^+}=\frac{1}{\sqrt{2}}(\ket{00}+\ket{11})$, when $\delta > 0$.

\begin{restatable}[equivalent to {\cite[Theorem 1]{BrassardMethotTapp2005:MinimumEntangledStatePseudoTelepathy}}]{theorem}{thmBipartite}\label{thm: bipartite}
Two-qubit states do not admit strongly non-local behaviour. 
\end{restatable}
\begin{proof}
This proof rests on defining an explicit global assignment $g:\LM\sqcup \LM\rightarrow O$ consistent with the possible events of the empirical model. More specifically, the map $g$ is obtained by assigning outcome $+1$ to one hemisphere of the Bloch sphere, and $-1$ to the other, with special conditions on the poles and a slight asymmetry between the two parties. 

We start by computing the amplitude $\braketTPv{\psi}$ of measuring $(\vtheta,\vphi)=\tuple{(\theta_1,\varphi_1),(\theta_2,\varphi_2)}$ on the general state \eqref{equ: bipartite} and obtaining joint outcome $\tuple{+1,+1}$:
\[
\braketTPv{\psi}=\cos\delta \cos\frac{\theta_1}{2}\cos\frac{\theta_2}{2}+ \sin\delta\sin\frac{\theta_1}{2}\sin\frac{\theta_2}{2}e^{-i(\varphi_1+\varphi_2)}
\]
Since $\delta=0$ gives rise to a product state, we will assume $\delta\neq 0$.

We define the following maps:
\begin{align*}
\fdecdef{g_1&}{\LM}{O}{(\theta,\varphi)}{\begin{cases}
+1 & \text{ if } \theta=\pi \text{ or } \left(\theta\neq 0 \text{ and }\varphi\in\left[-\frac{\pi}{2},\frac{\pi}{2}\right)\right)  \\
-1 & \text{ if }  \theta=0 \text{ or } \left(\theta\neq\pi \text{ and } \varphi\in \left[\frac{\pi}{2},\frac{3\pi}{2}\right) \right)
\end{cases}}\\
\fdecdef{g_2&}{\LM}{O}{(\theta,\varphi)}{\begin{cases}
+1 & \text{ if } \theta=\pi \text{ or } \left(\theta\neq 0 \text{ and }\varphi\in\left(-\frac{\pi}{2},\frac{\pi}{2}\right]\right)  \\
-1 & \text{ if } \theta=0 \text{ or } \left(\theta\neq\pi \text{ and } \varphi\in \left(\frac{\pi}{2},\frac{3\pi}{2}\right] \right)
\end{cases}}
\end{align*}
and let $\fdec{g\defeq g_1\sqcup g_2}{\LM\sqcup\LM}{O}$ be a global assignment. A graphical representation of the map $g$ can be found in Figure \ref{fig: Assignment_2}.
\begin{figure}[htbp]
\centering
\includegraphics[scale=0.5]{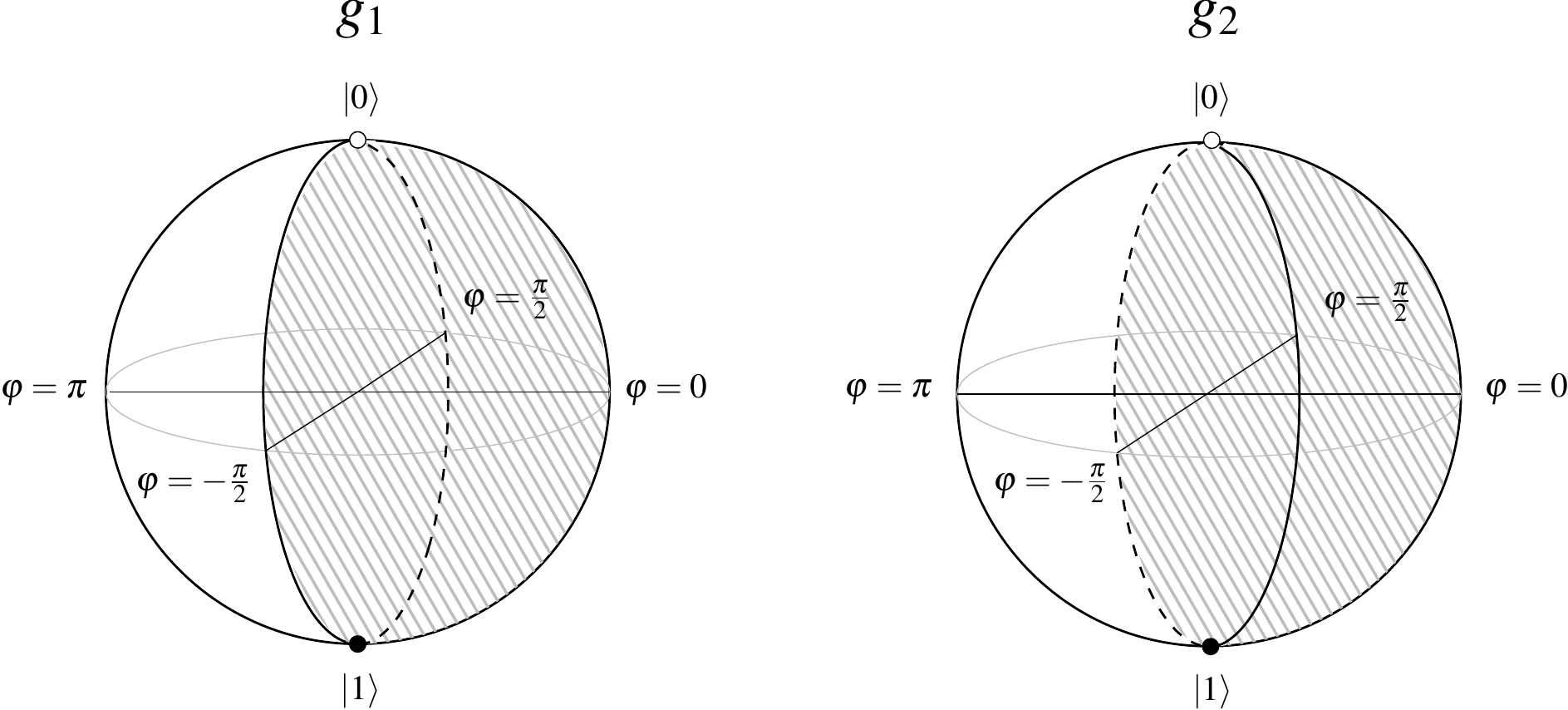}
\caption{Graphical representation of the global assignment $g$. The shaded region corresponds to the measurements mapped to $+1$ by $g$.}\label{fig: Assignment_2}
\end{figure}

Let $(\vtheta,\vphi)$ be a context whose individual measurements are mapped to $+1$ by $g$ (see Section~\ref{ssec:quantum} for why this is sufficient). In particular, it holds that $\theta_1,\theta_2\neq 0$. Since $\delta\neq 0$, we have 
\[
s\defeq\sin\delta \sin\frac{\theta_1}{2}\sin\frac{\theta_2}{2}>0
\Mand
c\defeq\cos\delta \cos\frac{\theta_1}{2}\cos\frac{\theta_2}{2}\ge 0.
\]
If $\theta_1=\pi$ or $\theta_2=\pi$, then $c=0$, which implies $\braketTPv{\psi}=se^{-i(\varphi_1+\varphi_2)}\neq 0$. Otherwise, $\varphi_1\in\left[-\frac{\pi}{2},\frac{\pi}{2}\right)$, $\varphi_2\in\left(-\frac{\pi}{2},\frac{\pi}{2}\right]$ and $\braketTPv{\psi}=c+se^{-i(\varphi_1+\varphi_2)}$ is the sum of a positive real number and a non-zero complex number. For it to be zero, the latter must be real and negative, hence 
\[
\varphi_1+\varphi_2=\pi\mod 2\pi,
\]
which cannot be satisfied in the domain of $\varphi_1,\varphi_2$.
\end{proof}

\section{W-SLOCC states are not strongly non-local}\label{sec:W}
A general state in the SLOCC class of the W state $\W=\frac{1}{\sqrt{3}}(\ket{001}+\ket{010}+\ket{100})$ can be written, up to LU, as
\begin{equation}\label{equ: W-SLOCC}
\WSLOCC =\sqrt{a}|001\rangle+\sqrt{b}|010\rangle+\sqrt{c}|100\rangle+\sqrt{d}|000\rangle,
\end{equation}
where $a,b,c\in\mathbb{R}_{>0}$ and $d\defeq 1-(a+b+c)\in\mathbb{R}_{\ge0}$. Indeed, we can obtain $\WSLOCC$ from $\W$ by applying the following ILO to $\W$:
\[
\left(
\begin{matrix}
\sqrt{a} & \sqrt{b}\\
0 & \sqrt{c}
\end{matrix}
\right)
\otimes
\left(
\begin{matrix}
\sqrt{3} & 0\\
0 & \frac{\sqrt{3b}}{\sqrt{a}}
\end{matrix}
\right)
\otimes
I.
\]

In order to prove that W-SLOCC states are not strongly non-local, we will need the following lemma, which generalises the argument used in the proof of Theorem~\ref{thm: bipartite} to show that the amplitude could not be zero. 

\begin{lemma}\label{lem: complex sum}
Let $z_1,\ldots,z_m\in\mathbb{C}$, and $r\in\mathbb{R}_{\ge0}$.
If 
\begin{equation}\label{equ: sum}
\sum_{i=1}^mz_i+r=0,
\end{equation}
then one of the following holds:
(i) $z_1=\cdots = z_m=r=0$;
(ii) there exists a $z_k\in\mathbb{R}_{< 0}$;
(iii) there exists $1\leq k, l\leq m$ such that \emph{$\mbox{Arg}(z_k)\in(0,\pi)$} and \emph{$\mbox{Arg}(z_l)\in(-\pi,0)$}.
\end{lemma}

\begin{proof}
If all the $z_i$ are real, then, since $r$ is non-negative, we must have either (i) or (ii). Now, suppose there is a $1\leq k\leq m$ such that $\mbox{Im}(z_k)\neq 0$. By \eqref{equ: sum}, we have $\sum_{i=1}^n\mbox{Im}(z_i)=0$. Thus,
\[
\sum_{i\neq k}\mbox{Im}(z_i)=-\mbox{Im}(z_k) \quad\bimplies\quad\sum_{i\neq k}|z_i|\sin(\mbox{Arg}(z_i))=-|z_k|\sin(\mbox{Arg}(z_k)).
\]
Hence, there exists at least one $l\neq k$ for which the sign of
$\mbox{Im}(z_l)$  
is opposite to that of
$\mbox{Im}(z_k)$, 
which implies that $z_l$ and $z_k$ are in different sides of the real axis, implying the condition about $\mbox{Arg}(z_l)$ and $\mbox{Arg}(z_k)$.
\end{proof}

\begin{restatable}{theorem}{thmWSLOCC}\label{thm: W-SLOCC}
States in the SLOCC class of W do not admit strongly non-local behaviour.
\end{restatable}

\begin{proof}
Similarly to the bipartite case of Theorem \ref{thm: bipartite}, the key idea of the proof is the definition of a global assignment $g:\LM\sqcup \LM\sqcup\LM\rightarrow O$ whose restriction to each context is contained in the support of the model. Once again, $g$ is obtained by partitioning the Bloch sphere into two hemispheres to which are assigned different outcomes, with asymmetric polar conditions across the parties.

We start by computing the amplitude $\braketTPv{\WSLOCCt }$ of measuring $(\vtheta,\vphi)$ on the general state \eqref{equ: W-SLOCC} and obtaining joint outcome $\tuple{+1,+1,+1}$:

\begin{equation}\label{equ: W amplitude}
\begin{split}
\braketTPv{\WSLOCCt }&= \underbrace{\sqrt{a}\left(\cos\frac{\theta_1}{2}\cos\frac{\theta_2}{2}\sin\frac{\theta_3}{2}
e^{-i\varphi_3}\right)}_{\equalscolon z_3\in\mathbb{C}}+\underbrace{\sqrt{b}\left(\cos\frac{\theta_1}{2}\cos\frac{\theta_3}{2}\sin\frac{\theta_2}{2}
e^{-i\varphi_2}\right)}_{\equalscolon z_2\in\mathbb{C}}\\
+ & \underbrace{\sqrt{c}\left(\cos\frac{\theta_2}{2}\cos\frac{\theta_3}{2}\sin\frac{\theta_1}{2}
e^{-i\varphi_1}\right)}_{\equalscolon z_1\in\mathbb{C}}+\underbrace{\sqrt{d}\left(\cos\frac{\theta_1}{2}\cos\frac{\theta_2}{2}\cos\frac{\theta_3}{2}
\right)}_{\equalscolon r\in\mathbb{R}_{\ge 0}}.
\end{split}
\end{equation}
Define the following functions:
\begin{align*}
\fdecdef{h=g_1=g_2&}{\LM}{O}{(\theta,\varphi)}{\begin{cases}
+1 & \text{ if } \theta=0 \text{ or } \left(\theta\neq \pi \text{ and } \varphi\in(-\pi,0]\right)\\
-1 & \text{ if } \theta=\pi \text{ or } \left(\theta\neq 0 \text{ and } \varphi\in(0,\pi]\right) 
\end{cases}}\\
\fdecdef{g_3&}{\LM}{O}{(\theta,\varphi)}{\begin{cases}
+1 & \text{ if } \theta=\pi \text{ or }\left(\theta\neq 0 \text{ and } \varphi\in(-\pi,0]\right)\\
-1 & \text{ if } \theta= 0 \text{ or } \left(\theta\neq \pi \text{ and } \varphi\in(0,\pi]\right)
\end{cases}}
\end{align*}
and let $\fdec{g\defeq h\sqcup h\sqcup g_3}{\LM\sqcup\LM\sqcup\LM}{O}$ be a global assignment. The map $g$ is graphically represented in Figure \ref{fig: Assignment_3}.
\begin{figure}[htbp]
\centering
\includegraphics[scale=0.5]{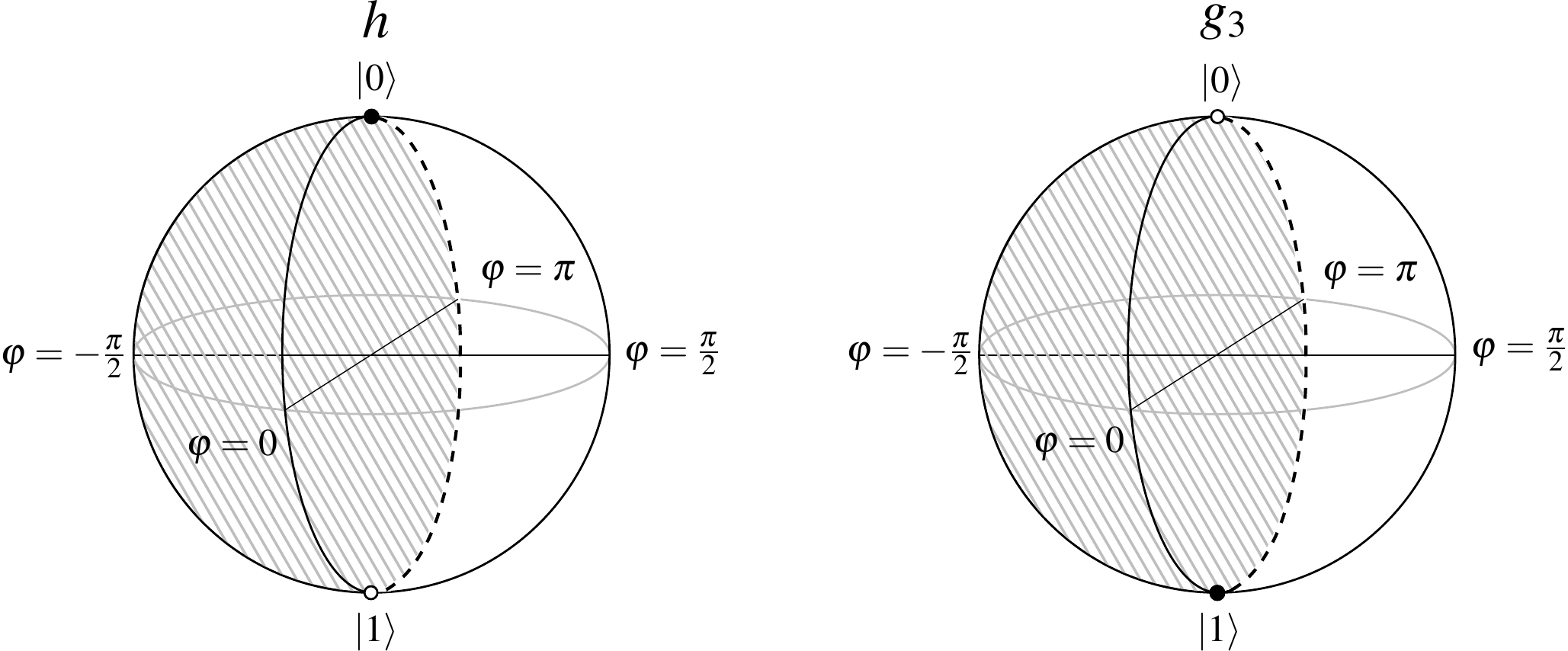}
\caption{Graphical representation of the global assignment $g$. The shaded region corresponds to the measurements mapped to $+1$ by $g$.}\label{fig: Assignment_3}
\end{figure}

Let $(\vtheta,\vphi)$ be a context whose individual measurements are mapped to $+1$ by $g$. In particular, $\theta_1,\theta_2\neq\pi$ and $\theta_3\neq 0$. Since $a>0$, we have 
\[
|z_3|=\sqrt{a}\cos\frac{\theta_1}{2}\cos\frac{\theta_2}{2}\sin\frac{\theta_3}{2}>0,
\]
which implies $z_3\neq 0$. Now, if $\theta_3=\pi$, then $z_1=z_2=r=0$ and $\braketTPv{\WSLOCCt }=z_3\neq 0$.

Otherwise, $\theta_3\neq \pi$ and $\varphi_3\in(-\pi,0]$, implying that $\mbox{Arg}(z_3)=-\varphi_3\in[0,\pi)$. For $i=1,2$, we either have $\theta_i=0$ or $\varphi_i\in(-\pi, 0]$, implying that $z_i=0$ or $\mbox{Arg}(z_i)=-\varphi_i\in[0,\pi)$. Using Lemma~\ref{lem: complex sum}, we conclude that $\braketTPv{\WSLOCCt}\neq 0$: (i) fails because $z_3\neq 0$, while (ii) and (iii) fail because $\mbox{Arg}(z_i)\in[0,\pi)$ whenever $z_i\neq 0$.
\end{proof}

\section{Strong non-locality in the SLOCC class of GHZ}\label{sec:GHZ}

\subsection{The $n$-partite GHZ state and local equatorial measurements}

Before we tackle the general case of GHZ-SLOCC states, we consider the GHZ state itself. We show that equatorial measurements are the only relevant ones in the study of strong non-locality for this state. In fact, this holds for the general $n$-partite GHZ state,
\[
\GHZn\defeq \frac{1}{\sqrt{2}}\left(\ket{0}^{\otimes n}+\ket{1}^{\otimes n}\right) \Mcomma
\]
and consequentely, in light of the remark towards the end of Section~\ref{ssec:quantum}, for any state in its LU class.
In the next section, we generalise this result to arbitrary states in the SLOCC class of the tripartite GHZ state,
and study conditions for strong non-locality within this class.

\begin{restatable}{theorem}{thmGHZn}\label{thm: GHZn}
Any strongly non-local behaviour of $\GHZn$ can be witnessed using only equatorial measurements. That is, there is a global assignment $g$ consistent with the model $e^{\GHZn}$ in all contexts that are not exclusively composed of equatorial measurements. 
\end{restatable}
\begin{proof}
The proof is achieved using a construction of a global assignment similar to the ones previously discussed. 

First, we derive the formula for the amplitude $\braketTPv{\GHZnt}$ of measuring $(\vtheta,\vphi)$ and obtaining joint outcome $\tuple{+1,\ldots,+1}$:

\begin{equation*}
\braketTPv{\GHZnt}= \frac{1}{\sqrt{2}} \left( \prod_{i=1}^n \cos\frac{\theta_i}{2} + e^{-i \sum_{i=1}^n \varphi_i} \prod_{i=1}^n \sin\frac{\theta_i}{2} \right).
\end{equation*}

Consider the function 
\[
\fdecdef{h}{\LM}{O}{(\theta,\varphi)}{\begin{cases}
+1 & \text{ if } \theta\in\left[0,\frac{\pi}{2}\right]\\
-1 & \text{ if } \theta\in \left(\frac{\pi}{2},\pi\right]
\end{cases}}
\]
i.e. $h$ assigns $+1$ to the equator and the northern hemisphere, and $-1$ to the southern hemisphere.
Let $\fdec{g\defeq \bigsqcup_{i=1}^n h}{\bigsqcup_{i=1}^n \LM}{O}$.
We show that this global assignment is consistent with the probabilities at all contexts that include at least a non-equatorial measurement. 

Let $(\vtheta, \vphi)$ be a context whose measurements are mapped to $+1$ by $g$. In particular, $\theta_i\leq\frac{\pi}{2}$ for all $i$. If $\braketTPv{\GHZnt}=0$, then
\[
 \prod_{i=1}^n \cos\frac{\theta_i}{2} = -e^{-i \left(\sum_{i=1}^n \varphi_i\right)} \prod_{i=1}^n \sin\frac{\theta_i}{2}
\]
Taking the modulus of both sides and dividing the right-hand by the left-hand side yields:
\[
\prod_{i=1}^n \tan\frac{\theta_i}{2} = 1
\]
which is verified if and only if $\theta_i=\frac{\pi}{2}$ for all $1\leq i \leq n$. 
\end{proof}

\subsection{Balanced GHZ-SLOCC states and local equatorial measurements}

A general state in the SLOCC class of the GHZ state can be written, up to LU, as
\begin{equation}\label{equ: GHZ class}
\GHZSLOCC=\sqrt{K}(\cos\delta\ket{000} + \sin\delta e^{i\Phi}\ket{\varphi_1}\ket{\varphi_2}\ket{\varphi_3}),
\end{equation}
where $K=(1+2\cos\delta\sin\delta\cos\alpha\cos\beta\cos\gamma\cos\Phi)^{-1}$, and
\[
\ket{\varphi_1} = \cos\alpha\ket{0}+\sin\alpha\ket{1},
\quad
\ket{\varphi_2} = \cos\beta\ket{0}+\sin\beta\ket{1},
\quad
\ket{\varphi_3} = \cos\gamma\ket{0}+\sin\gamma\ket{1},
\]
for some $\delta\in(0,\pi/4]$, $\alpha,\beta,\gamma\in (0,\pi/2]$, and $\Phi\in[0,2\pi)$. Indeed, $\GHZSLOCC$ is obtained from $\GHZ$ via the ILO
\[
\sqrt{2K}\left(
\begin{matrix}
\cos\delta & \sin\delta\cos\alpha e^{i\Phi}\\
0 & \sin\delta\sin\alpha e^{i\Phi}\\
\end{matrix}
\right)
\otimes
\left(
\begin{matrix}
1 & \cos\beta\\
0 & \sin\beta\\
\end{matrix}
\right)
\otimes
\left(
\begin{matrix}
1 & \cos\gamma\\
0 & \sin\gamma\\
\end{matrix}
\right).
\]

In order to prove the results of this section, it is convenient to describe $\GHZSLOCC$ in a slightly different form. By applying local unitaries, we can rewrite it as 
\begin{equation}\label{equ: GHZ-SLOCC}
\GHZSLOCC=\sqrt{K}(\cos\delta\ket{v_{\lambda_1}}\ket{v_{\lambda_2}}\ket{v_{\lambda_3}}+\sin\delta e^{i\Phi}\ket{w_{\lambda_1}}\ket{w_{\lambda_2}}\ket{w_{\lambda_3}}),
\end{equation}
where
\begin{equation}\label{equ: lambda1}
\ket{v_\lambda} = \bl{\lambda}{0} =\cos\frac{\lambda}{2}\ket{0}+\sin\frac{\lambda}{2}\ket{1},
\quad
\ket{w_\lambda} =\bl{\pi-\lambda}{0}=\sin\frac{\lambda}{2}\ket{0}+\cos\frac{\lambda}{2}\ket{1}
\end{equation}
for some $\lambda_i \in [0,\frac{\pi}{2})$, $i=1,2,3$.
The action of this LU can be thought of as choosing a new orthonormal basis for each qubit:
a graphical illustration of this process can be found in Figure~\ref{fig: rotation}.
\begin{figure}[htbp]
\centering
\includegraphics[scale=0.6]{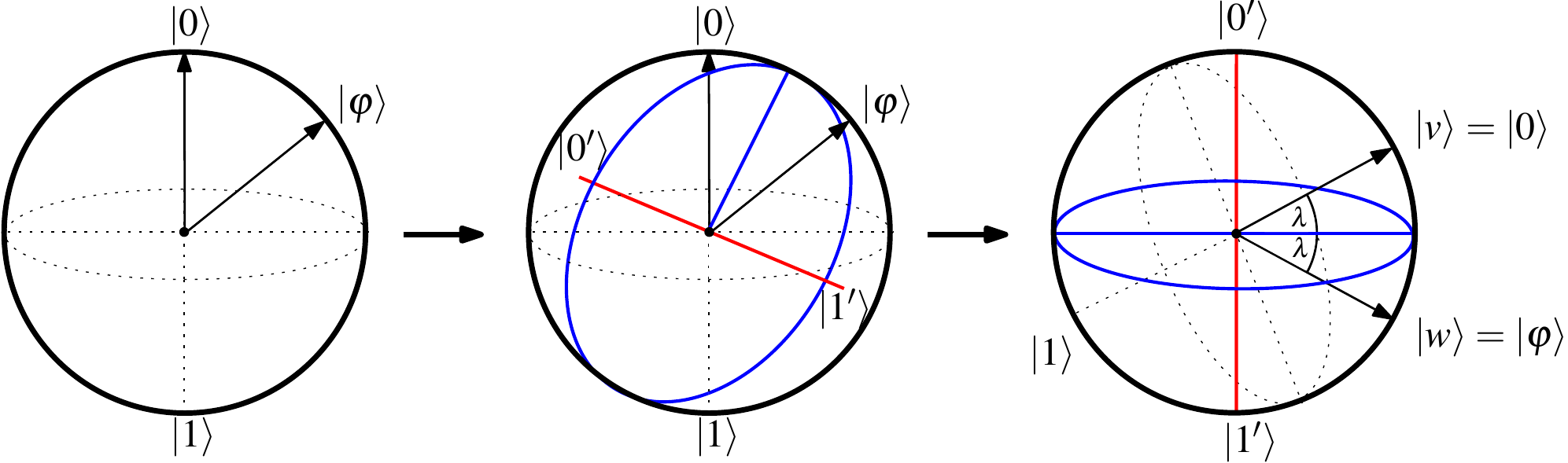}
\caption{Choice of a new basis $\enset{\ket{0'},\ket{1'}}$ for each qubit that allows the state to be described in the form \eqref{equ: GHZ-SLOCC}.}\label{fig: rotation}
\end{figure}
A key advantage of this LU-equivalent description of a general state in the GHZ SLOCC class
is that the equator of the $i$-th qubit's Bloch sphere coincides with
the great circle that bisects the $i$-th components of the two unique product states
that form a linear decomposition of the state.
Note that any state in the GHZ SLOCC class thus uniquely defines an equator in each Bloch sphere. It is to the measurements lying on these that we refer as being \emph{equatorial}.

We say that a state in the GHZ SLOCC class is $\emph{balanced}$
if the coefficients in its unique linear decomposition into a pair of product states have the same complex modulus --
when the state is written in the form \eqref{equ: GHZ-SLOCC}, this corresponds to having $\delta = \frac{\pi}{4}$, hence $\cos\delta = \sin\delta = \frac{1}{\sqrt{2}}$.

\begin{lemma}\label{lem: scalar product}
Let
$\ket{v_\lambda}$ and $\ket{w_\lambda}$ be given as in \eqref{equ: lambda1}, 
with $\lambda\in[0,\pi/2)$, and consider a measurement $(\theta,\varphi)$ with $\theta\in[0,\pi/2)$, i.e. with $+1$ eigenstate in the `northern hemisphere'.
Then $|\braket{\theta,\varphi}{v_\lambda}|>|\braket{\theta,\varphi}{w_\lambda}|$.
\end{lemma}

\begin{proof}
We have 
\[
\begin{split}
|\braket{\theta,\varphi}{v_\lambda}|>|\braket{\theta,\varphi}{w_\lambda}| &\bimplies \left| \cos\frac{\theta}{2}\cos\frac{\lambda}{2}+\sin\frac{\theta}{2}\sin\frac{\lambda}{2}e^{-i\varphi} \right| >\left| \cos\frac{\theta}{2}\sin\frac{\lambda}{2}+\sin\frac{\theta}{2}\cos\frac{\lambda}{2}e^{-i\varphi} \right|\\
&\bimplies \left|1+\tan\frac{\lambda}{2}\tan\frac{\theta}{2}e^{-i\varphi}\right|>\left|\tan\frac{\lambda}{2}+\tan\frac{\theta}{2}e^{-i\varphi}\right|,
\end{split}
\]
where, for the last step, we divide both sides by $\cos\frac{\lambda}{2}\cos\frac{\theta}{2}$, which is never $0$ since $\lambda,\theta\in[0,\pi/2)$. Let $x\defeq\tan\frac{\lambda}{2}$ and $y\defeq\tan\frac{\theta}{2}$, then
\[
\begin{split}
|1+xye^{-i\varphi}|>|x+ye^{-i\varphi}| &\Leftrightarrow |1+xy(\cos\varphi-i\sin\varphi)|>|x+y(\cos\varphi-i\sin\varphi)|\\
&\Leftrightarrow 1+2xy\cos\varphi+x^2y^2>x^2+2xy\cos\varphi+y^2\\
&\Leftrightarrow 1+x^2y^2-x^2-y^2>0\Leftrightarrow (1-x^2)(1-y^2)>0
\end{split}
\]
and this is always verified since $x,y\in[0,1)$ by the definition of the domains of $\theta$ and $\lambda$. 
\end{proof}

We use this lemma to generalise Theorem~\ref{thm: GHZn} to arbitrary states in the SLOCC class of the tripartite GHZ state.

\begin{restatable}{theorem}{thmEquatorialBalanced}\label{thm: equatorial balanced}
A state in the SLOCC class of GHZ that displays strong non-locality must be balanced. Moreover, any such strongly non-local behaviour can be witnessed using only equatorial measurements. 
\end{restatable}
\begin{proof}
The proof of this theorem can be derived by taking advantage of the special properties of balanced states and combining them with the argument used for Theorem \ref{thm: GHZn}. 

As before, we compute the amplitude $\braketTPv{\GHZSLOCCt}$:
\begin{equation*}
\braketTPv{\GHZSLOCCt}=\sqrt{K}\left(\cos\delta\prod_{i=1}^3\braketTPv{v_{\lambda_i}}+\sin\delta e^{i\Phi}\prod_{i=1}^3\braketTPv{w_{\lambda_i}}\right)
\end{equation*}
Take $\fdec{h}{\LM}{O}$ as defined in the proof of Theorem~\ref{thm: GHZn} and let $g\defeq h\sqcup h\sqcup h$. We claim that $g$ is consistent with the empirical probabilities at all contexts that include at least a non-equatorial measurement. 

Let $\TPv$ be a context whose measurements are all mapped to $+1$ by $g$. In particular, $\theta_i\leq\frac{\pi}{2}$ for $i=1,2,3$. If $\braketTPv{\GHZSLOCCt}=0$, then
\[
\cos\delta\prod_{i=1}^3\braketTPv{v_{\lambda_i}}=-\sin\delta e^{i\Phi}\prod_{i=1}^3\braketTPv{w_{\lambda_i}},
\]
and taking the complex modulus of both sides,
\[
\cos\delta\prod_{i=1}^3|\braketTPv{v_{\lambda_i}}|=\sin\delta\prod_{i=1}^3|\braketTPv{w_{\lambda_i}}|
\]
Since $\delta\in(0,\pi/4]$ we have $\cos\delta\ge\sin\delta$, with equality iff $\delta=\frac{\pi}{4}$. By Lemma~\ref{lem: scalar product}, we conclude that this equation can only be satisfied if $\delta=\frac{\pi}{4}$ (i.e. the state is balanced) and $\theta_i=\frac{\pi}{2}$ for $i=1,2,3$ (i.e. all the measurements are equatorial).
\end{proof}

\subsection{Further restrictions}

The theorem above allows us to reduce the scope of our search for strongly non-local behaviour in the SLOCC class of GHZ to: (i) balanced states, i.e. those of the form
\[
\Bal\defeq \sqrt{\frac{K}{2}}(\ket{v_{\lambda_1}}\ket{v_{\lambda_2}}\ket{v_{\lambda_3}}+e^{i\Phi}\ket{w_{\lambda_1}}\ket{w_{\lambda_2}}\ket{w_{\lambda_3}}),
\]
determined by a tuple $\vlambda=\tuple{\lambda_1,\lambda_2,\lambda_3}\in \left[0,\frac{\pi}{2}\right)^3$ and a phase $\Phi$,
where $\ket{v_\lambda}$ and $\ket{w_\lambda}$ are given as in \eqref{equ: lambda1};
(ii) local equatorial measurements in the sense defined above, i.e. those with $+1$ eigenstate
\[
\ket{\varphi}\defeq\left|\frac{\pi}{2},\varphi\right\rangle=\frac{1}{\sqrt{2}}(\ket{0}+e^{i\varphi}\ket{1})
\]
for $\varphi\in[0,2\pi)$.
Given this premise, we are interested in understanding when the amplitude function $\braket{\vphi}{\Balt}$ is $0$. We have:
\begingroup
\allowdisplaybreaks
\begin{align}
\braket{\vphi}{\Balt}=0 &\bimplies \prod_{i=1}^3\braket{\varphi_i}{v_{\lambda_i}}+e^{i\Phi}\prod_{i=1}^3\braket{\varphi_i}{w_{\lambda_i}}=0 \nonumber \\
&\bimplies \prod_{i=1}^3\braket{\varphi_i}{w_{\lambda_i}} = -e^{-i\Phi}\prod_{i=1}^3\braket{\varphi_i}{v_{\lambda_i}}\nonumber \\
&\bimplies \prod_{i=1}^3\braket{\varphi_i}{w_{\lambda_i}} = -e^{-i\Phi}\prod_{i=1}^3e^{-i\varphi_i}\overline{\braket{\varphi_i}{w_{\lambda_i}}}\label{equ: conjugate}\\
&\bimplies \prod_{i=1}^3e^{i\varphi_i}\braket{\varphi_i}{w_{\lambda_i}}\overline{\braket{\varphi_i}{w_{\lambda_i}}}^{^{-1}}=-e^{-i\Phi}\nonumber\\
&\bimplies\prod_{i=1}^3e^{i\varphi_i}\left(\frac{\braket{\varphi_i}{w_{\lambda_i}}}{|\braket{\varphi_i}{w_{\lambda_i}}|}\right)^2=-e^{-i\Phi}\nonumber\\
&\bimplies \sum_{i=1}^3\left(\varphi_i+2\mbox{Arg}\braket{\varphi_i}{w_{\lambda_i}}\right)=\pi-\Phi\mod 2\pi\nonumber
\end{align}
\endgroup
where to get \eqref{equ: conjugate} we use 
\[
\braket{\varphi}{v_{\lambda}}
=
\frac{1}{\sqrt{2}}\left(\cos\frac{\lambda}{2}+\sin\frac{\lambda}{2}e^{-i\varphi}\right)
=
\frac{e^{-i\varphi}}{\sqrt{2}}\left(\cos\frac{\lambda}{2}e^{i\varphi}+\sin\frac{\lambda}{2}\right)
=
e^{-i\varphi}\overline{\braket{\varphi}{w_{\lambda}}}.
\]
and for the last step we take the argument of two complex numbers of norm $1$. Defining
\[
\beta(\lambda, \varphi)\defeq\varphi+ 2\mbox{Arg}\braket{\varphi}{w_\lambda}=\varphi-2\arctan\left(\frac{\sin\frac{\lambda}{2}\sin\varphi}{\cos\frac{\lambda}{2}+\sin\frac{\lambda}{2}\cos\varphi}\right),
\]
we can rewrite the condition above as
\begin{equation}\label{equ: vanishing}
\braket{\vphi}{\Balt}=0 \quad\bimplies\quad \sum_{i=1}^3\beta(\lambda_i,\varphi_i)=\pi-\Phi \mod 2\pi
\end{equation}

\begin{restatable}{proposition}{propLambda}\label{prop: lambda}
If $\lambda_1+\lambda_2+\lambda_3>\frac{\pi}{2}$, the state $\Balzero$ does not admit strongly non-local behaviour. 
\end{restatable}
\begin{proof}
We start by showing that the map $\beta(\lambda,\varphi)$, seen as a function of $\varphi$, is strictly increasing for all $\lambda\in\left[0,\frac{\pi}{2}\right)$. To see this, it is sufficient to compute the derivative:
\[
\Forall{\lambda\in\left[0,\frac{\pi}{2}\right),\varphi\in[0,2\pi)}
\;\;
\frac{\partial}{\partial \varphi}\beta(\lambda,\varphi)=\frac{\cos\lambda}{1+\cos\varphi\sin\lambda} \Mdot
\]
This is strictly positive since $\cos\lambda>0$ and $\cos\varphi\sin\lambda>-1$ since $0\leq \sin\lambda <1$.

Now, define a function $\fdec{h}{[0,2\pi)}{O}$ by
\[
h(\varphi)\defeq \begin{cases}
+1 & \text{ if } \varphi\in\left(-\frac{\pi}{2},\frac{\pi}{2}\right]\\
-1 & \text{ if } \varphi\in\left(\frac{\pi}{2},\frac{3\pi}{2}\right]
\end{cases}
\]
and let $g\defeq h\sqcup h\sqcup h$. Take a context $\vphi$ whose measurements are assigned $+1$ by $g$, i.e. $\varphi_i\in\left(-\frac{\pi}{2},\frac{\pi}{2}\right]$. Using the fact that $\beta(\lambda,-)$ is increasing, we have
\[
\left|\sum_{i=1}^3\beta(\lambda_i,\varphi_i)\right|\leq \sum_{i=1}^3|\beta(\lambda_i,\varphi_i)|\leq \sum_{i=1}^3\beta\left(\lambda_i,\frac{\pi}{2}\right)=\sum_{i=1}^3\left(\frac{\pi}{2}-\lambda_i\right)=\frac{3\pi}{2}-\sum_{i=1}^3\lambda_i<\frac{3\pi}{2}-\frac{\pi}{2}=\pi.
\]
Consequently, $\sum_{i=1}^3\beta(\lambda_i,\varphi_i)\neq \pi \mod 2\pi$, hence by \eqref{equ: vanishing},
$\braket{\vphi}{\Balzerot}\neq0$ as required.
\end{proof}

\section{A family of strongly non-local three-qubit models}\label{sec:GHZ-family}
\begin{restatable}{theorem}{thmGHZfamily}\label{thm: GHZ-family}
Let $m\in\mathbb{N}_{>0}$ and $N\defeq 2m$ an even number.
Consider the tripartite measurement scenario with $X_1=X_2=\enset{0,\ldots,N-1}$ and $X_3=\enset{0,\frac{N}{2}}$. The empirical model determined by the state $\Balarg{\tuple{0,0,\lambda_N}}{0}$, where $\lambda_N\defeq \frac{\pi}{2}-\frac{\pi}{N}$, with the measurement label $i$ at each site interpreted as the local equatorial measurement $\cos \frac{i\pi}{N}\sigma_X+\sin \frac{i\pi}{N}\sigma_Y$ (i.e. the measurement with $+1$ eigenstate $\bl{\frac{\pi}{2}}{i\frac{\pi}{N}}$), is strongly non-local.
\end{restatable}

\begin{proof}
This proof rests on deriving, using the algebraic structure of $\ZZ_{2N}$,
a (conditional) system of linear equations over $\mathbb{Z}_{2}$ that must be satisfied by any global assignment consistent with the possible events of the empirical model, yet does not admit any solution.
This seems to be closely related to the general concept of all-vs-nothing (AvN) arguments introduced in \cite{AbramskyEtAl:ContextualityCohomologyAndParadox},
but does not quite fit this setting.
The reason is that the system of linear equations that a global assignment $g$ must satisfy
depends on the value that $g$ assigns to a particular measurement.
In that sense, this could be seen as a conditional version of an AvN argument.

Consider a context $\tuple{i,j,k}\in X_1\times X_2\times X_3$, with $i,j\in\enset{0,\ldots,N-1}$, $k \in \enset{0,m}$,
and a triple of outcomes $\tuple{a_i,b_j,c_k}\in\mathbb{Z}_2^3$ for the measurements in the context.\footnote{For this proof, it is convenient to relabel $+1, -1, \times$ as $0,1,\oplus$, where $\oplus$ denotes addition modulo $2$.} 
From equation \eqref{equ: vanishing}, we know that measuring $\tuple{i,j,k}$ and obtaining outcomes $\tuple{a_i,b_j,c_k}$ has probability zero if and only if 
\begin{equation}\label{equ: =0}
\beta\left(0,i\frac{\pi}{N}+a_i\pi\right)+\beta\left(0,j\frac{\pi}{N}+b_j\pi\right)+\beta\left(\frac{\pi}{2}-\frac{\pi}{N},k\frac{\pi}{N}+c_k\pi\right)=\pi \mod 2\pi
\end{equation}
With simple computations, we can show that $\beta(0,\varphi)=\varphi$ for all $\varphi\in[0,2\pi)$, and that
\begin{equation}\label{equ: c}
\beta\left(\frac{\pi}{2}-\frac{\pi}{N},c_0\pi\right)=c_0\pi\quad\text{and}\quad\beta\left(\frac{\pi}{2}-\frac{\pi}{N},\frac{\pi}{2}+c_m\pi\right)=(-1)^{c_m}\frac{\pi}{N}.
\end{equation}

An arbitrary global assignment is defined by choosing outcomes for all the measurements in $X_1 \sqcup X_2 \sqcup X_3$:
\[
a_0,\ldots, a_{N-1}, b_0,\ldots, b_{N-1}, c_0, c_m\in\mathbb{Z}_2.
\]
By \eqref{equ: =0} and \eqref{equ: c}, such an assignment is consistent with the probabilities of the empirical model at every context if and only if
\begin{equation*}
\begin{cases}
i\frac{\pi}{N}+a_i\pi+j\frac{\pi}{N}+b_j\pi+c_0\pi\neq\pi  &\mod2\pi\quad \forall i,j\in\enset{0,\ldots,N-1}\\
i\frac{\pi}{N}+a_i\pi+j\frac{\pi}{N}+b_j\pi+(-1)^{c_m}\frac{\pi}{N}\neq\pi  & \mod 2\pi\quad \forall i,j\in\enset{0,\ldots,N-1}
\end{cases}
\end{equation*}
We will proceed to show that this system admits no solution, which implies strong non-locality.
By identifying the group $\setdef{k\frac{\pi}{N}}{k\in\mathbb{Z}_{2N}}$ with $\mathbb{Z}_{2N}$, we can equivalently rewrite
\begingroup
\allowdisplaybreaks
\begin{align*}
&\begin{cases}
i+a_iN+j+b_jN+c_0N\neq N  & \mod 2N\quad \forall i,j\\
i+a_iN+j+b_jN+(-1)^{c_m}\neq N & \mod 2N\quad \forall i,j
\end{cases}\\
\bimplies&\\&
\begin{cases}
i+j+N(a_i\oplus b_j\oplus c_0)\neq N & \mod 2N\quad \forall i,j\\
i+j+(-1)^{c_m}+N(a_i\oplus b_j)\neq N & \mod 2N\quad \forall i,j
\end{cases}\\
\bimplies&\\&
\begin{cases}
a_i\oplus b_j\oplus c_0=0 & \forall i,j \text{ s.t. } i+j=0\\
a_i\oplus b_j\oplus c_0=1 & \forall  i,j \text{ s.t. } i+j=N\\
& \\
a_i\oplus b_j=0 & \forall  i,j \text{ s.t. } i+j+(-1)^{c_m}=0\\
a_i\oplus b_j=1 & \forall i,j \text{ s.t. } i+j+(-1)^{c_m}=N.
\end{cases}\\
\bimplies&\\&
\begin{cases}
a_0\oplus b_0\oplus c_0=0 & \\
a_i\oplus b_{N-i}\oplus c_0=1 & \forall i \text{ s.t. }1\leq i\leq N-1\\
& \\
a_i\oplus b_{N-i-1}=1 & \forall i \text{ s.t. } 0\leq i\leq N-1 \quad\quad\text{ if } c_m=0\\
& \\
a_0\oplus b_1=0 & \\
a_1\oplus b_0=0 &  \hphantom{\forall i \text{ s.t. } 0\leq i\leq N-1} \quad\quad\text{ if } c_m=1\\
a_i\oplus b_{N+1-i}=1 & \forall i \text{ s.t. }2\leq i\leq N-1
\end{cases}
\end{align*}%
\endgroup
Since $N=2m$ is even, if we sum all the $N$ equations from the first two lines we obtain
\[
\bigoplus_{i=0}^{N-1}a_i\oplus\bigoplus_{j=0}^{N-1}b_j=1.
\]
On the other hand, if we sum any of the other two groups of $N$ equations we get
\[
\bigoplus_{i=0}^{N-1}a_i\oplus \bigoplus_{j=0}^{N-1}b_j=0,
\]
showing that the system is unsatisfiable regardless of whether $c_m=0$ or $c_m=1$. 
\end{proof}

This new family of strongly non-local three-qubit systems is tightly connected to a construction on two-qubit states due to Barrett, Kent, and Pironio \cite{BarrettKentPironio2006:MaximallyNonlocalAndMonogamousQuantumCorrelations}. In particular, our empirical models restricted to the first two parties coincide, up to a rotation of the equatorial measurements, to those used in \cite{BarrettKentPironio2006:MaximallyNonlocalAndMonogamousQuantumCorrelations}. The local fraction of these bipartite empirical models tends to zero as the number of measurements increases, but obviously none of them are strongly non-local. Despite the lack of strong non-locality in the bipartite systems constructed in \cite{BarrettKentPironio2006:MaximallyNonlocalAndMonogamousQuantumCorrelations}, we show that it is possible to witness strongly non-local behaviour with a finite amount of measurements by adding a third qubit with some entanglement, and only two local measurements -- Pauli $X$ and $Y$ -- available on it.
An interesting aspect is that there is a trade-off between the number of measuring settings available on the first two qubits and the amount of entanglement between the third qubit and the system comprised of the other two.

We illustrate this by computing the bipartite von Neumann entanglement entropy between the first two qubits and the third, i.e. the von Neumann entropy of the reduced state of $\Balarg{\tuple{0,0,\lambda}}{0}$ corresponding to the third qubit, as a function of $\lambda$.
Let $\rho_{ABC}$ denote the density matrix of $\Balarg{\tuple{0,0,\lambda}}{0}$. The reduced density matrix corresponding to the third qubit is
\[
\rho_C(\lambda)=\Tr_{AB}[\rho_{ABC}]=\braket{00|_{_{AB}}\rho_{ABC}}{00}_{_{AB}}+\braket{11|_{_{AB}}\rho_{ABC}}{11}_{_{AB}}=\frac{1}{2}
\left(
\begin{matrix}
1 & 2\cos\frac{\lambda}{2}\sin\frac{\lambda}{2}\\
2\cos\frac{\lambda}{2}\sin\frac{\lambda}{2} & 1
\end{matrix}
\right).
\]
The eigenvalues of $\rho_C(\lambda)$ are $\epsilon_{\pm}(\lambda)\defeq \frac{1}{2}(1\pm\sin\lambda)$. Hence, by rewriting $\rho_C(\lambda)$ in its eigenbasis, we can easily compute the von Neumann entropy $S_C$ as a function of $\lambda$:
\[
S_C(\lambda)\defeq -\Tr\left[ \rho_C(\lambda)\log_2\rho_C(\lambda)\right]=-\epsilon_+(\lambda)\log_2\epsilon_+(\lambda)  -   \epsilon_-(\lambda)\log_2\epsilon_-(\lambda)
\]
The plot of the function $S_C(\lambda)$ is shown in Figure \ref{fig: trade-off}. 
\begin{figure}[htbp]
\centering
\includegraphics[scale=0.55]{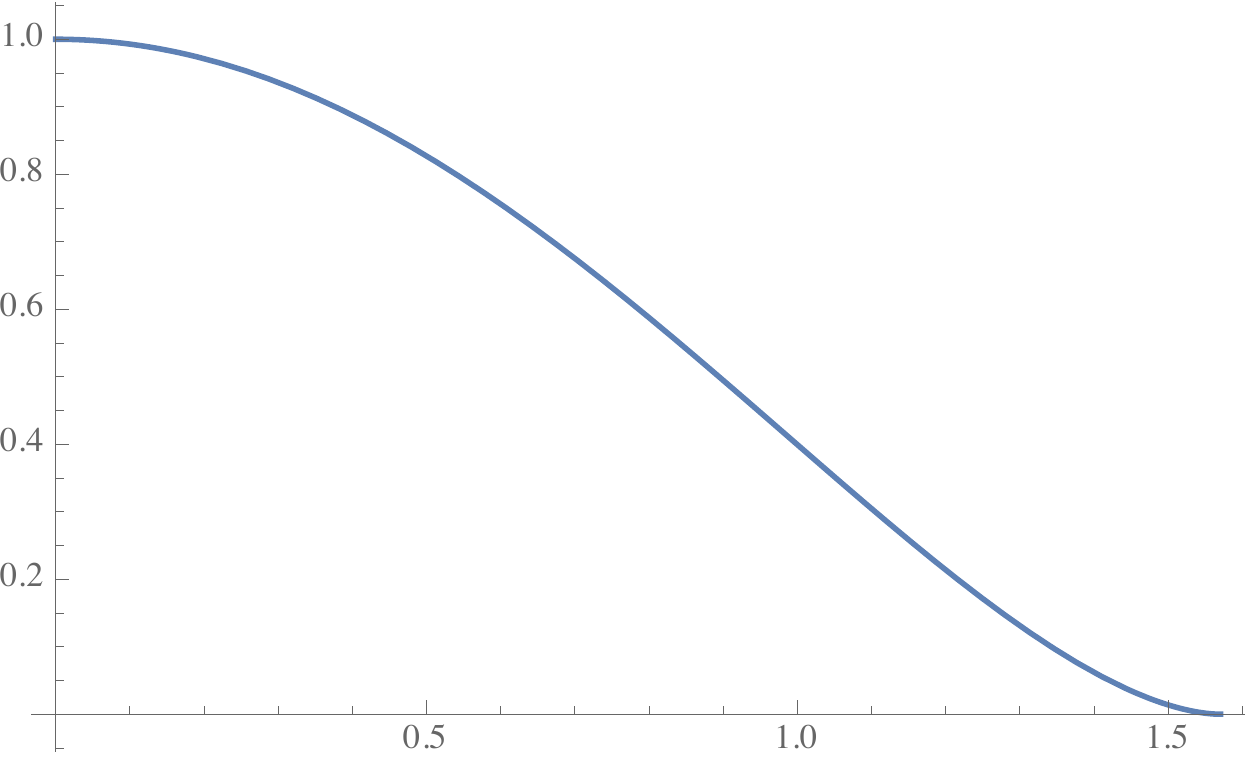}
\caption{Von Neumann entanglement entropy between the third qubit of $\Balarg{\tuple{0,0,\lambda}}{0}$ and the other two as a function of $\lambda$.}\label{fig: trade-off}
\end{figure}
Notice that the entanglement entropy is maximal, i.e. equal to $1$, when $N=2$, in which case $\lambda_2=0$ and so $\Balarg{\tuple{0,0,\lambda_2}}{0}=\GHZ$. This corresponds to the usual GHSZ argument with Pauli measurements $X$,$Y$ for each qubit. On the other hand, $S(\lambda)$ becomes arbitrarily small as $N \to \infty$, when $\lambda_N \to \frac{\pi}{2}$ and  $\Balarg{\tuple{0,0,\lambda_N}}{0}$ approaches the state $\BellState\otimes \ket{+}$, which has no entanglement between the first two qubits and the third.

\section{Outlook}\label{sec:outlook}

Our analysis of strong non-locality for three-qubit systems has been quite extensive. We shall discuss a number of directions for further research.

\begin{enumerate}
\item First, it remains to complete our classification of all instances of three-qubit strong non-locality.
\item
The original GHSZ--Mermin model witnesses the yet stronger algebraic notion of all-versus-nothing (AvN) non-locality, formalised in a general setting in \cite{AbramskyEtAl:ContextualityCohomologyAndParadox}, and indeed provides one of the motivating examples for considering this kind of non-locality.
The family of strongly non-local models introduced in Section~\ref{sec:GHZ-family} does not fit this framework exactly. 
Nevertheless, our proof of strong non-locality does make essential use of the algebraic structure of $\mathbb{Z}_{2N}$ (or the circle group), in what amounts to a conditional version of an AvN argument.
One may wonder whether a similar property will hold for all instances of three-qubit strong non-locality.
\item This family also highlights an inter-relationship between non-locality, entanglement and the number of measurements available, and raises the question of whether this is an instance of a more general relationship. 
\item Finally, while the present results provide necessary conditions for strong non-locality in three-qubit states, the more general question of characterising strong non-locality of $n$-qubit states, where little is known about SLOCC classes, remains open.
\end{enumerate}

\section*{Acknowledgments} This work was carried out in part while some of the authors visited the Simons Institute for the Theory of Computing (supported by the Simons Foundation) at the University of California, Berkeley, as participants of the Logical Structures in Computation programme (AB, RSB, GC, NdS, SM),
and while SM was based at the Institut de Recherche en Informatique Fondamentale, Universit\'e Paris Diderot -- Paris 7.
Support from the following is also gratefully acknowledged: EPSRC EP/N018745/1 (SA, RSB) and EP/N017935/1 (NdS), `Contextuality as a Resource in Quantum Computation'; EPSRC Doctoral Training Partnership and Oxford--Google Deepmind Graduate Scholarship (GC); U.S. AFOSR FA9550-12-1-0136, `Topological \& Game-Semantic Methods for Understanding Cyber Security' (KK); Fondation Sciences Math\'ematiques de Paris, post-doctoral research grant otpFIELD15RPOMT-FSMP1, `Contextual Semantics for Quantum Theory' (SM).

\bibliographystyle{eptcs}
\bibliography{refs_tripartite}

\end{document}